\documentclass[a4paper,11pt]{article}
\usepackage{graphicx}
\usepackage{amsfonts}
\usepackage{amsmath}
\usepackage{hyperref}
\title{How not to secure wireless sensor networks: A plethora of insecure polynomial-based key pre-distribution schemes}
\author{Chris J. Mitchell\\Information Security Group, Royal Holloway, University of London\\
\url{www.chrismitchell.net}}
\date{5th October 2020 (version 5)}
\newtheorem{lemma}{Lemma}[section]
\newtheorem{corollary}[lemma]{Corollary}

\newenvironment{proof}[1][Proof]{\begin{trivlist}
\item[\hskip \labelsep {\bfseries #1}]}{\end{trivlist}}
\newcommand{\qed}{\nobreak \ifvmode \relax \else
      \ifdim\lastskip<1.5em \hskip-\lastskip
      \hskip1.5em plus0em minus0.5em \fi \nobreak
      \vrule height0.75em width0.5em depth0.25em\fi}
\begin{document}

\maketitle

\section*{Abstract}

Three closely-related polynomial-based group key pre-distribution schemes have recently been
proposed, aimed specifically at wireless sensor networks.  The schemes enable any subset of a
predefined set of sensor nodes to establish a shared secret key without any communications
overhead. It is claimed that these schemes are both secure and lightweight, i.e.\ making them
particularly appropriate for network scenarios where nodes have limited computational and storage
capabilities. Further papers have built on these schemes, e.g.\ to propose secure routing protocols
for wireless sensor networks. Unfortunately, as we show in this paper, all three schemes are
completely insecure; whilst the details of their operation varies, they share common weaknesses. In
two cases we show that an attacker equipped with the information built into just one sensor node
can compute all possible group keys, including those for which the attacked node is not a member;
this breaks a fundamental design objective. In the other case an attacker equipped with the
information built into at most two sensor nodes can compute all possible group keys. In the latter
case the attack can also be achieved by an attacker armed with the information from a single node
together with a single group key to which this sensor node is not entitled. Repairing the schemes
appears difficult, if not impossible. The existence of major flaws is not surprising given the
complete absence of any rigorous proofs of security for the proposed schemes. A further recent
paper proposes a group membership authentication and key establishment scheme based on one of the
three key pre-distribution schemes analysed here; as we demonstrate, this scheme is also insecure,
as the attack we describe on the corresponding pre-distribution scheme enables the authentication
process to be compromised.

\section{Introduction}  \label{section-intro}

In this paper we are concerned with the problem of enabling groups of nodes (such as exist in a
Wireless Sensor Network (WSN)) to establish shared secret keys for use by any subset of the nodes,
using only information distributed to the nodes in advance by a trusted \emph{Key Generation Centre
(KGC)}. In particular we examine three separate (albeit closely related) protocols described as
being designed for wireless sensor networks, although they have no particular features restricting
their operation to this use case. Indeed, they would work just as effectively for any case where a
number of hardware tokens, such as RFID tags or smart cards, are to be distributed by a single
trusted party.  The schemes we consider are designed to ensure that nodes which are part of the
system, but not part of a selected group, cannot access the key for that group, i.e.\ nodes can
only access keys for groups of which they are a part.

All three of the schemes we consider are polynomial-based.  They appear to be inspired by the work
of Blundo et al.\ \cite{Blundo98}; however, unlike this prior art, the schemes all have major
flaws. Unfortunately, many other polynomial-based group key distribution schemes have been shown to
be flawed, \cite{Liu17,Mitchell18,Mitchell19a,Mitchell19,Mitchell20}, meaning that the approach
needs to be used with care. In any event, there are many existing schemes which achieve the same
objectives in an efficient way and which have rigorous proofs of security --- see, for example,
Boyd et al.\ \cite{Boyd20}, and, of course, the previous work of Blundo et al.\ \cite{Blundo98}.

We also briefly examine a group membership and key establishment scheme recently proposed by Cheng
et al.\ \cite{Cheng20a}.  This scheme builds directly on the first of the schemes we cryptanalyse
in this paper, and we show how the cryptanalysis extends to this case.

The remainder of this paper is structured as follows.  Some preliminary remarks are made in
\S\ref{section-prelims}. The 2015 Harn-Hsu protocol is described in \S\ref{section-harn-hsu},
together with an attack against this scheme; the fact that this attack also works on the derived
2020 Cheng-Hsu-Xia-Harn scheme is also briefly discussed. The very closely-related 2015 Harn-Gong
scheme is briefly considered in \S\ref{section-harn-gong}. A more recent scheme, the 2019
Albakri-Harn protocol, is described and analysed in \S\ref{section-albakri}.  Other issues are
discussed in \S\ref{section-other}, and brief conclusions are provided in
\S\ref{section-conclusions}.

\section{A preliminary observation}  \label{section-prelims}

All three of the schemes we describe involve arithmetic computed in the ring of integers
$\mathbb{Z}_N$, where $N$ is an `RSA modulus' \cite{Rivest78}, i.e.\ $N=pq$ for two large prime
numbers $p$ and $q$. This means that addition and multiplication are computed modulo $N$.  For
today's RSA applications, $p$ and $q$ are typically chosen to be of the order of $2^{1024}$, and we
make this assumption throughout. This ensures that with today's techniques and computing resources,
factoring $N$ is infeasible (and this is necessary to ensure that the RSA cryptosystem is secure).

If $p$ and $q$ are of this size, then the probability that a random integer $s$ ($1\leq s<n$) is
coprime to $N=pq$ is very close to 1.  To see this, observe that $\phi(n)$, the Euler phi function
that gives the number of positive integers less than $n$ that are coprime to $n$ (Menezes et al.\
\cite{Menezes97}, Definition 2.100), satisfies $\phi(pq)=(p-1)(q-1)$ for any primes $p$, $q$
(\cite{Menezes97}, Fact 2.101).  Thus the probability that a random positive integer less than
$N=pq$ (where $p$ and $q$ are of the order of $2^{1024}$) is coprime to $N$ is
\[ \phi(N)/N = (p-1)(q-1)/pq \approx (2^{2048}-2^{1025})/2^{2048} = 1-2^{-1023}. \]
Hence we assume throughout that all randomly chosen integers are coprime to $N$ and so have
multiplicative inverses modulo $N$; moreover, such inverses can readily be computed
(\cite{Menezes97}, Algorithm 2.142). That is, we can readily compute divisions modulo $N$, a fact
that underlies all the attacks we describe.

\section{The Harn-Hsu scheme}  \label{section-harn-hsu}

In 2015 Harn and Hsu \cite{Harn15} proposed a group key pre-distribution protocol for WSNs.
Unfortunately, as we describe in detail below, the Harn-Hsu scheme is insecure. In particular,
anyone with the shares belonging to one node can compute the keys for all possible groups,
regardless of membership. The nature of the flaws giving rise to the attack are rather fundamental,
meaning that it is difficult to imagine how the scheme could be rescued.

\subsection{The protocol}  \label{subsection-harn-hsu-operation}

We next describe the Harn-Hsu protocol \cite{Harn15}.  The protocol involves a set of $\ell$ end
nodes $\mathcal{S}=\{S_1,S_2,\ldots,S_\ell\}$ and a Key Generation Centre (KGC).  Note that, in
this description and in the description of related protocols (in \S\ref{section-harn-gong} and
\S\ref{section-albakri}), the terminology and notation have been changed slightly from that used in
the source papers to ensure consistency.  As with all such protocols, the Harn-Hsu protocol has two
phases, as follows.

\begin{itemize}
\item \emph{Share generation}, a pre-distribution phase in which the KGC generates a set of
    $\ell-1$ shares $\{s_{i,1}(x), s_{i,2}(x), \ldots, s_{i,\ell-1}(x)\}$ for each node $S_i
    \in \mathcal{S}$, $1\leq i\leq\ell$, where each share $s_{i,j}(x)$ is a polynomial. It is
    assumed that the KGC distributes the shares to the nodes using a secure channel, e.g.\ at
    the time of manufacture/personalisation, and that each node stores its shares securely.
\item \emph{Group key establishment}, where all members of a subset of nodes independently
    generate a \emph{group key} for use by nodes in that subset. Note that, if repeated with
    the same subset of nodes, the key establishment process will generate the same key, but a
    different key is generated for each node subset.
\end{itemize}
We next describe the operation of these two phases in greater detail.

\subsubsection{Share generation}

The KGC first chooses an `RSA modulus' $N$, i.e.\ $N=pq$ where $p$ and $q$ are two primes chosen to
be sufficiently large to make factoring $N$ infeasible.  It is also assumed that each node $S_i$
has a unique identifer $\mbox{ID}_i$, $1\leq i\leq\ell$, where $1\leq \mbox{ID}_i< N$.  It is
(presumably) the case that the KGC keeps $p$ and $q$ secret, although, as we see below, this does
not appear to make a significant difference to the security of the scheme.

The KGC next chooses a polynomial $f$ over $\mathbb{Z}_N$ of degree $k$ for some pre-chosen $k$.
We suppose that the polynomial coefficients are chosen uniformly at random from $\mathbb{Z}_N$. No
explicit guidance is given on the choice of $k$, but later it is claimed that the scheme is secure
if up to $k$ nodes are captured and their secrets are revealed.  We assume throughout that
$k\geq2$.

For every $S_i$ ($1\leq i\leq\ell$) the KGC calculates a set of $\ell-1$ shares in the following
way.
\begin{itemize}
\item The KGC first computes $f(\mbox{ID}_i)$.
\item The KGC then chooses at random $\ell-2$ values $u_{i,1},u_{i,2},\ldots,u_{i,\ell-2}$,
    where $0<u_{i,j}<N$ for every $j$ ($1\leq j\leq \ell-2$).
\item The KGC then computes $u_{i,\ell-1}$ as
\[ u_{i,\ell-1}= \frac{f(\mbox{ID}_i)}{\prod_{j=1}^{\ell-2}u_{i,j}} \bmod N \]
i.e.
\[ \prod_{j=1}^{\ell-1}u_{i,j} \equiv f(\mbox{ID}_i) \pmod N. \]
\item Finally, the shares for $S_i$ are computed as $s_{i,j}(x)=u_{i,j}f(x) \bmod N$, $1\leq
    j\leq \ell-1$, where each share is a polynomial over $\mathbb{Z}_N$ of degree $k$.
\end{itemize}

Finally, the KGC equips each node $S_i\in\mathcal{S}$ with the following:
\begin{itemize}
\item the shares $\{s_{i,1}(x), s_{i,2}(x), \ldots, s_{i,\ell-1}(x)\}$;
\item the values of $N$ and the node's own identifier $\mbox{ID}_i$;
\item the identifiers $\mbox{ID}_j$ of all of the other nodes.
\end{itemize}

\subsubsection{Group key establishment}  \label{subsubsection-harn-hsu-gke}

Suppose some subset $\mathcal{S}'\subseteq \mathcal{S}$ of sensors ($|\mathcal{S}'|\geq 2$) wish to
share a group key, where $|\mathcal{S}'|=h$ ($2\leq h\leq\ell$). The group key
$K_{\mathcal{S}'}\in\mathbb{Z}_N$ for $\mathcal{S}'$ is defined to be:
\[ K_{\mathcal{S}'} = f(0)^{\ell-h}\prod_{S_j\in\mathcal{S}'}f(\mbox{ID}_j) \bmod
N. \]

Any node $S_i\in \mathcal{S}'$ can compute $K_{\mathcal{S}'}$ using the shares $\{s_{i,1}(x),
s_{i,2}(x), \ldots, s_{i,\ell-1}(x)\}$ in the following way.
\begin{itemize}
\item Suppose that
    \[\mathcal{S}'-\{S_i\}=\{S_{r_1},S_{r_2},\ldots,S_{r_{h-1}}\}\]
    and that
    \[\mathcal{S}-\mathcal{S}'-\{S_i\}=\{S_{r_h},S_{r_{h+1}},\ldots,S_{r'_{\ell-1}}\}\]
\item Then set
\[ K_{\mathcal{S}'} = \prod_{j=1}^{h-1}s_{i,j}(\mbox{ID}_{r_j}) \prod_{j=h}^{\ell-1}s_{i,j}(0). \]
\end{itemize}

However, a node $S_k\not\in\mathcal{S}'$ cannot use its shares to compute $K_{\mathcal{S}'}$, at
least not in the way described above.

Finally note that it should be clear that the group key $K_{\mathcal{S}}$ for the set of all
sensors is simply
\[ K_{\mathcal{S}} = \prod_{j=1}^{\ell}f(\mbox{ID}_j) \bmod
N. \]

\subsection{Critical vulnerabilities}  \label{subsection-harn-hsu-attack}

We describe below a simple insider attack which enables the discovery of all the group keys. This
can be achieved using the shares held by any one of the nodes.  The attack works regardless of the
choice of the polynomial degree $k$.  We also show that, even without access to a set of shares,
knowledge of some keys enables others to be deduced.

\subsubsection{Some simple deductions}

We start with a simple but key observation where here, and in the remainder of
\S\ref{section-harn-hsu}, we suppose that the values $z_1,z_2,\ldots,z_\ell$ satisfy
\[ f(\mbox{ID}_r) \equiv z_rf(0) \pmod N \]
and $0\leq z_r<N$ for every $r$ ($1\leq r\leq\ell$).

\begin{lemma}  \label{harn-hsu-lemma0}
If $s_{i,j}(x)$ is a share for user $S_i$ ($1\leq i\leq\ell, 1\leq j\leq \ell-1$), then, for any
$r$ ($1\leq r\leq\ell$):
\[ z_r = s_{i,j}(\mbox{ID}_r)/s_{i,j}(0) \bmod N. \]
\end{lemma}

\begin{proof}
By definition, we immediately have:
\begin{eqnarray*}
s_{i,j}(\mbox{ID}_r)/s_{i,j}(0) & \equiv &
u_{i,j}f(\mbox{ID}_r)/u_{i,j}f(0) \pmod N\\
& = & f(\mbox{ID}_r)/f(0).
\end{eqnarray*}

The result follows immediately from the definition of $z_r$.\qed
\end{proof}

We next have a related result.

\begin{lemma}  \label{harn-hsu-lemma1}
If $\mathcal{S}'\subseteq \mathcal{S}$ is some non-empty subset of the nodes, then the key
$K_{\mathcal{S}'}$ for the group $\mathcal{S}'$ satisfies:
\[ K_{\mathcal{S}'} = f(0)^\ell\prod_{S_j\in\mathcal{S}'}z_j  \bmod N. \]
\end{lemma}

\begin{proof}
By definition (and using the notation established in \S\ref{subsubsection-harn-hsu-gke}) we have
\[ K_{\mathcal{S}'} = f(0)^{\ell-h}\prod_{S_j\in\mathcal{S}'}f(\mbox{ID}_j)  \bmod
N \] where $h=|\mathcal{S}'|$. But we assumed that $f(\mbox{ID}_r) \equiv z_rf(0) \pmod N$, and hence
\[ K_{\mathcal{S}'} = f(0)^{\ell-h}\prod_{S_j\in\mathcal{S}'}z_jf(0)  \bmod
N \] and the result follows by re-arranging the products. \qed
\end{proof}

This immediately gives the following.

\begin{corollary}  \label{harn-hsu-corollary1}
The group key $K_{\mathcal{S}}$ shared by all nodes satisfies
\[K_{\mathcal{S}} = f(0)^\ell \prod_{j=1}^{\ell}z_j \bmod N. \]
\end{corollary}

We also have the following.

\begin{corollary}  \label{harn-hsu-corollary2}
Suppose $\mathcal{S}_1=\mathcal{S}_2\cup\{S_i\}$ for some $S_i$ ($1\leq i\leq \ell$).  Then
\[ K_{\mathcal{S}_1}/K_{\mathcal{S}_2} \equiv z_i \pmod N. \]
\end{corollary}

\subsubsection{Completing the insider attack}

First observe that, using Lemma~\ref{harn-hsu-lemma0}, anyone possessing the shares belonging to a
single sensor node can learn the values of $z_r$, $1\leq r\leq \ell$, for every $r$.

Next observe that, since a single set of shares enables recovery of the group key $K_{\mathcal{S}}$
shared by all nodes, the attacker can obtain $f(0)^\ell \bmod N$ since it follows from
Corollary~\ref{harn-hsu-corollary1} that:
\[ f(0)^\ell = \frac{K_{\mathcal{S}}}{\prod_{j=1}^{\ell}z_j} \bmod N. \]

Given knowledge of $f(0)^\ell \bmod N$ together with the complete set of values
$z_1,z_2,\ldots,z_\ell$, Lemma~\ref{harn-hsu-lemma1} enables the computation of any group key
$K_{\mathcal{S}'}$.  This completes the attack.

\subsubsection{An outsider attack}

Even if an attacker does not possess any of the shares, then attacks are still possible if the
attacker has access to some of the keys $K_{\mathcal{S}'}$ generated using the system.  We give a
simple example of an attack, but many other variants are possible.

Suppose an attack has access to three group keys: $K_{\mathcal{S}_1}$, $K_{\mathcal{S}_2}$ and
$K_{\mathcal{S}_3}$, for groups $\mathcal{S}_1$, $\mathcal{S}_2$ and $\mathcal{S}_3$, where the
attacker also knows the membership of these groups.  Suppose also (to make the discussion simpler)
that $\mathcal{S}_1 = \mathcal{S}_2\cup\{S_y\}$ for some node $S_y$.  Then, by
Corollary~\ref{harn-hsu-corollary2}, the attacker can immediately compute $z_y$.

If $S_y\in\mathcal{S}_3$, then the attacker can now compute
$K_{\mathcal{S}_4}=K_{\mathcal{S}_3}/z_r \bmod N$, where $\mathcal{S}_4=\mathcal{S}_3-\{S_y\}$.
That is, the attacker can compute another valid group key.

This attack works because there are simple algebraic relationships between the keys.

\subsection{Breaking the Cheng-Hsu-Xia-Harn scheme}

As noted in the introduction, Cheng et al.\ \cite{Cheng20a} recently proposed a group membership
authentication and key establishment scheme; as noted by the authors the `proposed protocol is
built upon a recent paper by Harn and Hsu', i.e.\ precisely the paper we have analysed immediately
above.

The Cheng-Hsu-Xia-Harn scheme involves a trusted KGC generating and distributing shares to all
members of a group, exactly as in the Harn-Hsu scheme.  When a subset of the group wishes to
conduct group authentication, each subset member generates a shared key using precisely the method
described in the Harn-Hsu paper --- they then use this key to authenticate to the members of the
subset. However, as we have shown, every `subset key' is actually computable by every shareholder,
regardless of whether they are in the authorised subset or not; this means that the authentication
scheme can be trivially broken.

\section{The Harn-Gong scheme}  \label{section-harn-gong}

In 2015, the same year in which the Harn-Hsu paper appeared, Harn and Gong \cite{Harn15a} presented
another group key pre-distribution scheme.  It would appear that this paper was actually submitted
a few months earlier than the Harn-Hsu paper.  However, we consider it after Harn-Hsu as the scheme
it describes is essentially just a special case of the Harn-Hsu scheme.

Like the Harn-Hsu scheme, a KGC distributes shares to each of a set of $\ell$ nodes.  The KGC first
chooses an RSA modulus $N$ and a polynomial $f(x)$ over $\mathbb{Z}_N$ of degree $k$ for some
pre-chosen $k$.  The single share sent to each user $S_i$ is computed as $s_i(x)=u_if(x) \bmod N$.
where $u_i$ satisfies
\[ (u_i)^{\ell-1} \equiv f(\mbox{ID}_i) \pmod N. \]

Using the notation of \S\ref{section-harn-hsu}, it should be clear that this is simply a special
case of the Harn-Hsu scheme where the set of $\ell-1$ shares for a user are all chosen to be equal,
i.e.\ where $u_{i,1}=u_{i,2}=\ldots=u_{i,\ell-1}$ and hence
$s_{i,1}(x)=s_{i,2}(x)=\ldots=s_{i,\ell-1}(x)$.  All other aspects of the scheme are identical.

Since it is just a special case of the Harn-Hsu scheme, precisely the same attacks work, and hence
we do not consider the scheme further here.

\section{The Albakri-Harn scheme}  \label{section-albakri}

Four years after the Harn-Hsu paper appeared, in 2019 Albakri and Harn \cite{Albakri19} proposed
yet another group key pre-distribution protocol for WSNs.  Three variants of the protocol are
described, a \emph{basic scheme}, which has a heavy storage overhead, and two derived schemes which
use the same underlying idea but reduce the storage requirement for individual nodes.  Given its
fundamental role, we focus here on the basic scheme.

At first sight the scheme is somewhat different to the two schemes we have examined so far.
However, more careful analysis reveals that it is again very closely related.  Not surprisingly,
the Albakri-Harn scheme is also insecure.  In particular, if two nodes collude, or if one node
gains access to a single key to which it is not entitled, then all keys for all groups, regardless
of membership, can be computed.

\subsection{The protocol}  \label{subsection-albakri-operation}

We describe the \emph{Basic scheme}, \cite{Albakri19}.  The protocol involves a set of $\ell$
sensor nodes $\mathcal{S}=\{S_1,S_2,\ldots,S_\ell\}$ and the KGC.  The protocol has two phases.
\begin{itemize}
\item \emph{Token generation}, a pre-distribution phase in which the KGC generates a token
    $T_i$ for each node $S_i \in \mathcal{S}$, $1\leq i\leq\ell$.  As previously, it is assumed
    that the KGC distributes the tokens to the nodes using a secure channel, e.g.\ at the time
    of manufacture/personalisation.
\item \emph{Group key establishment}, where all members of a subset of nodes independently
    generate a \emph{group key} for use by nodes in that subset.
\end{itemize}
We next describe the operation of these two phases in greater detail.

\subsubsection{Token generation}

Again as before, the KGC first chooses an `RSA modulus' $N$, i.e.\ $N=pq$ where $p$ and $q$ are two
primes chosen to be sufficiently large to make factoring $N$ infeasible.  It is also assumed that
each node $S_i$ has a unique identifer $\mbox{ID}_i$, $1\leq i\leq\ell$, where $1\leq \mbox{ID}_i<
N$.

The KGC next chooses a set of $\ell$ univariate polynomials $\mathcal{F}=\{f_1,f_2,\ldots,f_\ell\}$
over $\mathbb{Z}_N$, each of degree $t-1$ for some $t$.  In the absence of further information, we
suppose here the coefficients of these polynomials are chosen uniformly at random from
$\mathbb{Z}_N$. No explicit guidance is given on the choice of $t$, but later it is claimed that
the scheme is secure if up to $t-1$ nodes are captured (and presumably their secrets are revealed),
and that the choice of $t$ is a trade-off between the computational complexity of key establishment
and the security of the scheme. We assume throughout that $t\geq2$, since if $t=1$ all polynomials
are of degree zero and all tokens (and group keys) are identical.

For every $i$ ($1\leq i\leq\ell$) the KGC calculates the token $T_i$ as the following polynomial in
$\ell-1$ variables:
\[ T_i = f_i(\mbox{ID}_i)\prod_{\substack{j=1\\j\neq i}}^{\ell}f_j(x_j) \bmod N. \]

Finally, the KGC equips each node $S_i\in\mathcal{S}$ with the following:
\begin{itemize}
\item the token $T_i$;
\item the values of $N$ and the node's own identifier $\mbox{ID}_i$;
\item the identifiers $\mbox{ID}_j$ of all of the other nodes.
\end{itemize}

It is important to note that some of the steps above are based on the author's interpretation of
the Albakri-Harn paper, \cite{Albakri19}, as many details are left unclear.

\subsubsection{Group key establishment}

Suppose some subset $\mathcal{S}'\subseteq \mathcal{S}$ of sensors ($|\mathcal{S}'|\geq 2$) wish to
share a group key. The group key $K_{\mathcal{S}'}\in\mathbb{Z}_N$ for $\mathcal{S}'$ is defined to
be:
\[ K_{\mathcal{S}'} = \prod_{S_j\in\mathcal{S}'}f_j(\mbox{ID}_j) \prod_{S_v\in\mathcal{S}-\mathcal{S}'}f_v(0) \bmod
N. \]

Any of the nodes $S_i\in \mathcal{S}'$ can compute $K_{\mathcal{S}'}$ by evaluating the token $T_i
\bmod N$ for a particular choice of the indeterminates $x_j$, namely by setting

\[ x_j= \left\{ \begin{array}{ll}
        \mbox{ID}_j & \mbox{if $S_j\in\mathcal{S}'$} \\
        0 & \mbox{otherwise.}
        \end{array}
        \right. \]

However, a node $S_k\not\in\mathcal{S}'$ cannot use its token $T_k$ to compute $K_{\mathcal{S}'}$,
at least not in the way described above.

Finally note that it should be clear that the group key $K_{\mathcal{S}}$ for the set of all
sensors is simply
\[ K_{\mathcal{S}} = \prod_{j=1}^{\ell}f_j(\mbox{ID}_j) \bmod
N. \]

\subsubsection{An observation}

It was stated above that this scheme is closely related to the Harn-Hsu scheme.  This is not
immediately apparent, as the Harn-Hsu scheme involves a node being given a set of shares each of
which is a univariate polynomial, and in the Albakri-Harn scheme a sensor node is given a token
consisting of a single multivariate polynomial.  However, this token is actually analogous to the
product of the Harn-Hsu shares, in the case where the share polynomials have distinct
indeterminates (although the polynomials $f_i$ in Albakri-Harn are all distinct). Armed with this
insight, the protocol then works in an essentially identical way to Harn-Hsu. It is therefore not
surprising that, as we discuss below, closely analogous attacks apply.

\subsection{An attack}  \label{subsection-albakri-attack}

\subsubsection{Stage 1: Partial polynomial recovery}  \label{subsubsection-albakri-stage1}

In the discussions below we need notation for the coefficients of the polynomials in $\mathcal{F}$,
and hence, recalling that all these polynomials have degree $t-1$, we suppose that
\[ f_i(x) = \sum_{j=0}^{t-1} f_{ij}x^j \]
for every $i$, $1\leq i\leq\ell$.

We first consider what a single node $S_i$ can learn about the polynomials in $\mathcal{F}$ from a
single token $T_i$.  It follows immediately from the definition that $T_i$ consists of the sum of
all terms of the form
\[ f_i(\mbox{ID}_i)\prod_{\substack{j=1\\j\neq i}}^{\ell}f_{jk_j}x_j^{k_j} \bmod N, \]
where $0\leq k_j\leq t-1$ for every $j$.  That is, $S_i$ will know the value of
\[ f_i(\mbox{ID}_i)\prod_{\substack{j=1\\j\neq i}}^{\ell}f_{jk_j} \bmod N \]
for every pair $(j,k_j)$, where $1\leq j\leq\ell$ ($j\neq i$) and $0\leq k_j\leq t-1$.

Given these observations, for any $r$ ($1\leq r\leq\ell$, $r\neq i$) and any $s$ ($1\leq s$), and
using the known coefficients of the polynomial $T_i$, it is simple to compute the ratio of two such
coefficients, namely
\[ \frac {f_i(\mbox{ID}_i)\prod_{\substack{j=1\\j\neq i}}^{\ell}f_{jk_j}}
         {f_i(\mbox{ID}_i)\prod_{\substack{j=1\\j\neq i}}^{\ell}f_{jk'_j}}
         \bmod N \]
where $k_j=k'_j$ for every $j$ except for $j=r$ and $k_j=s$, where we put $k'_r=0$.  In this case
all the identical terms will cancel, and the above expression will equal
\[ \frac {f_{rs}}
         {f_{r0}}
         \bmod N. \]
That is, for every polynomial $f_r$ ($r\neq i$), the ratio of each coefficient in $f_r$ to the
constant coefficient $f_{r0}$ can be computed.  Hence $S_i$ can readily discover the values
$w_1,w_2,\ldots,w_{t-1}$ where
\[ f_{rj}=w_jf_{r0}. \]
Moreover, given $\mbox{ID}_r$ (which is known to all nodes), $S_i$ can use the above to deduce that
\[ f_r(\mbox{ID}_r) = z_rf_{r0} \]
for some $z_r$ known to $S_i$, for every $r\neq i$.

\subsubsection{Stage 2: Pair-wise collusion to complete the attack}  \label{subsubsection-albakri-stage2}

We next describe how, if two nodes collude, e.g.\ by sharing their tokens, they can completely
break the system; that is they will have the means to readily compute every possible group key,
including for groups excluding both of them.  That is, as is the case for all three schemes
examined, the system is completely insecure if just two nodes collude, regardless of the choice of
$t$.

We first need the following simple result (analogous to Lemma~\ref{harn-hsu-lemma1}), which uses
the notation of \S\ref{subsection-albakri-operation}.

\begin{lemma}  \label{lemma1}
Suppose the values $z_1,z_2,\ldots,z_\ell$ satisfy
\[ f_r(\mbox{ID}_r) \equiv z_rf_{r0} \pmod N \]
for every $r$ ($1\leq r\leq\ell$). If $\mathcal{S}'\subseteq \mathcal{S}$ is some non-empty subset
of the nodes, the key $K_{\mathcal{S}'}$ for the group $\mathcal{S}'$ satisfies:
\[ K_{\mathcal{S}'} = \prod_{S_j\in\mathcal{S}'}z_j \prod_{v=1}^{\ell}f_{v0} \bmod
N. \]
\end{lemma}

\begin{proof}
By definition we have
\[ K_{\mathcal{S}'} = \prod_{S_j\in\mathcal{S}'}f_j(\mbox{ID}_j) \prod_{S_v\in\mathcal{S}-\mathcal{S}'}f_v(0) \bmod
N. \]

We assumed that $f_r(\mbox{ID}_r) = z_rf_{r0} \bmod N$ for every $r$, and it trivially holds that
$f_r(0)=f_{r0}$ for every $r$.  Hence
\[ K_{\mathcal{S}'} = \prod_{S_j\in\mathcal{S}'}z_jf_{j0} \prod_{S_v\in\mathcal{S}-\mathcal{S}'}f_{v0} \bmod
N. \] The result now follows by re-arranging the products. \qed
\end{proof}

The following corollary (analogous to Corollary~\ref{harn-hsu-corollary1}) is immediate.

\begin{corollary}  \label{corollary1}
The group key $K_{\mathcal{S}}$ shared by all nodes satisfies
\[K_{\mathcal{S}} = \prod_{j=1}^{\ell}z_j \prod_{v=1}^{\ell}f_{v0} \bmod N. \]
\end{corollary}

We now observe that, from the arguments in \S\ref{subsubsection-albakri-stage1}, if any two users
collude then they can learn the complete set of values $z_1,z_2,\ldots,z_\ell$ satisfying
$f_r(\mbox{ID}_r) = z_rf_{r0} \bmod N$, $1\leq r\leq\ell$.

Next observe that, since they can both readily compute the group key $K_{\mathcal{S}}$ shared by
all nodes, they can obtain $\prod_{v=1}^{\ell}f_{v0} \bmod N$ since it follows from
Corollary~\ref{corollary1} that:
\[ \prod_{v=1}^{\ell}f_{v0} = \frac{K_{\mathcal{S}}}{\prod_{j=1}^{\ell}z_j} \bmod N. \]

Armed with knowledge of $\prod_{v=1}^{\ell}f_{v0} \bmod N$ together with the complete set of values
$z_1,z_2,\ldots,z_\ell$, the colluding nodes can use Lemma~\ref{lemma1} to immediately compute any
group key $K_{\mathcal{S}'}$, regardless of whether either of the colluding nodes are members of
the group $\mathcal{S}'$.  This completes the attack.

\subsubsection{An alternative to collusion}  \label{subsection-alternative}

Just as previously, even in this case it remains possible to completely break the system if a node
(by some means) learns a single key for a group of which the node is not a member.  Suppose node
$S_i$ knows the key $K_{\mathcal{S}'}$ for the group $\mathcal{S}'$, where $S_i\not\in
\mathcal{S}'$.  By Lemma~\ref{lemma1}, $S_i$ knows that:
\[ K_{\mathcal{S}'} = \prod_{S_j\in\mathcal{S}'}z_j \prod_{v=1}^{\ell}f_{v0} \bmod
N. \]

But $S_i$ knows all the values $z_j$ for $S_j\in\mathcal{S}'$, since $S_i\not\in \mathcal{S}'$.
Hence $S_i$ can compute $\prod_{v=1}^{\ell}f_{v0} \bmod N$ as
\[ \prod_{v=1}^{\ell}f_{v0} = \frac{K_{\mathcal{S}'}}{\prod_{S_j\in\mathcal{S}'}z_j} \bmod N. \]

Knowledge of $\prod_{v=1}^{\ell}f_{v0} \bmod N$ together with the (almost) complete set of values
$z_1,z_2,\ldots,z_\ell$ (except for $z_i$), enables $S_i$ to now compute any group key
$K_{\mathcal{S}'}$ for any group $\mathcal{S}'$ for which $S_i\not\in \mathcal{S}'$.  This
completes the attack.

\subsection{A simplified attack}  \label{subsection-albakri-simpler}

Using the same notation as before, we now present an even simpler version of the attack approaches
described in \S\ref{subsection-albakri-attack}.  This attack does not seek to recover information
about the polynomials in $\mathcal{F}$, but instead recovers just sufficient information to be able
to recover all group keys.  We use the notation established in \S\ref{subsection-albakri-attack}.

\subsubsection{Learning a set of ratios}

We start with another simple observation, analogous to Corollary~\ref{harn-hsu-corollary2}.

\begin{lemma}  \label{lemma2}
Suppose $\mathcal{S}_i=\mathcal{S}-S_i$.  Then
\[ K_{\mathcal{S}}/K_{\mathcal{S}_i} = z_i \bmod N \]
for every $i$, $1\leq i\leq \ell$.
\end{lemma}

\begin{proof}
By Lemma~\ref{lemma1},
\[ K_{\mathcal{S}_i} = \prod_{\substack{j=1\\j\neq i}}^{\ell}z_j\prod_{v=1}^{\ell}f_{v0} \bmod N.\]
Similarly, by Corollary~\ref{corollary1},
\[ K_{\mathcal{S}} = \prod_{j=1}^{\ell}z_j \prod_{v=1}^{\ell}f_{v0} \bmod N. \]
The result follows immediately. \qed
\end{proof}

From Lemma~\ref{lemma2}, it follows that $S_i$ can learn the values of $z_r$, $1\leq r\leq \ell$,
for every $r\neq i$.

\subsubsection{Completing the attack}

Armed with the values of $z_r$, $1\leq r\leq \ell$, for every $r\neq i$, the attacks described in
\S\ref{subsubsection-albakri-stage2} and \S\ref{subsection-alternative} now work exactly as
previously described.

\section{Other observations}  \label{section-other}

\subsection{But what about the security analyses?}

On the face of it, the above analyses are rather surprising in the light of apparently robust
claims made in the Harn-Hsu \cite{Harn15}, Harn-Gong \cite{Harn15a}, and Albakri-Harn
\cite{Albakri19} papers. In all three papers there are apparently `theorems' proving the security
of the respective schemes.  However, in every case closer examination reveals that the `proofs' are
not in any way rigorous arguments.

For example, Theorem 1 of Albakri and Harn \cite{Albakri19} states that `The adversaries cannot
obtain any information of secret polynomials selected by KGC'. The `secret polynomials' referred to
here are the polynomials in $\mathcal{F}$, and the set of possible adversaries includes parties
with knowledge of one or more tokens $T_i$. However, we have shown in
\S\ref{subsection-albakri-attack} that knowledge of a single token is sufficient to learn all but
one of the polynomials in $\mathcal{F}$ up to a constant term. How can this be, given Theorem 1?
Examination of the `proof' of Theorem 1 reveals that it is by no means a rigorous proof --- it is
more a series of unsubstantiated claims.  For example, the proof starts with the following
statement `\emph{Capturing one sensor} --- It is obvious that by capturing any single sensor $S_i$,
and obtaining the token $T_i$, the adversary cannot recover information of any individual
polynomial $f_i$, nor the product of all individual polynomials'. That is, the `proof' of the main
claim seems to amount to a statement that it is `obvious'. Sadly, the claimed result is clearly not
as obvious as the authors hoped.

More seriously, this highlights the need for newly proposed cryptographic protocols to be provided
with robust and rigorous evidence of their security.  Indeed, this has been the state of the art
for a couple of decades, as has been very widely discussed --- see, for example, Boyd et al.
\cite{Boyd20}.

\subsection{Three almost identical schemes}

Quite apart from the lack of security, it is unfortunate that, given all three papers we considered
(\cite{Albakri19,Harn15a,Harn15}) share one author, that three such similar schemes have been
published separately.  Moreover, the authors make no attempt to explain the close relationships
between the three schemes.

\subsection{An application}

To make matters worse, some authors have sought to build broader security schemes on top of one of
the schemes considered here.  For example, Harn et al.\ \cite{Harn16} describe a secure routing
protocol for WSNs, of which the Harn-Hsu scheme \cite{Harn15} forms an integral part; indeed, for
some reason the authors have chosen to describe the Harn-Hsu scheme again in some detail.  This, of
course, means that the routing protocol, regardless of its design, is inherently insecure.  It
would, of course, have been good design practice to describe the routing protocol in terms of a
`black box' technique for key pre-distribution, and then to mention possible candidates for this
black box, since there is no inherent reason to couple the two techniques. This would have avoided
the main problem.

\section{Summary and conclusions}  \label{section-conclusions}

As we have demonstrated, the Harn-Hsu, Harn-Gong and Albakri-Harn schemes all possess fundamental
flaws (as does the derived Cheng-Hsu-Xia-Harn scheme). Given the nature of these flaws, it is
difficult to imagine how the schemes could be rescued. Indeed, there is no reason to believe that a
secure scheme can be designed using the underlying approach adopted in all three cases. As observed
in \S\ref{section-intro}, many related polynomial-based group key distribution schemes have been
shown to be flawed, \cite{Liu17,Mitchell18,Mitchell19a,Mitchell19,Mitchell20}. Again as observed
above, there are many existing schemes which achieve the same objectives in an efficient way and
which have rigorous proofs of security --- see, for example, Boyd et al.\ \cite{Boyd20} and Blundo
et al.\ \cite{Blundo98}.

Fundamentally, the fact that the authors have not provided rigorous proofs of security for the
various schemes means that attacks such as those described here remain possible.  It would have
been more prudent to follow established wisdom and only publish a scheme of this type if a rigorous
security proof had been established.  Similar remarks apply to the all-too-often misconceived
attempts to fix broken schemes, unless a proof of security can be devised for a revised scheme.
Achieving this seems very unlikely for variants of the schemes we have examined.


\end{document}